\newcommand{\N}{\mathbb{N}}
\newcommand{\R}{\mathbb{R}}

\let \Ra \Rightarrow

\documentclass[runningheads]{llncs}
\usepackage[utf8]{inputenc}
\usepackage{amsfonts, amsmath, amssymb} 
\PassOptionsToPackage{hyphens}{url}\usepackage{hyperref}

\begin{document}

\title{On Separation between the Degree of a Boolean Function and the Block Sensitivity\thanks{The article was prepared within the framework of the HSE University Basic Research Program}}
\titlerunning{On Separation between the Degree and the Block Sensitivity}
\author{Nikolay V. Proskurin}
\authorrunning{N. V. Proskurin}
\institute{HSE University, Russian Federation\\
\email{nproskurin@hse.ru}}

\maketitle

\begin{abstract}
In this paper we study the separation between two complexity measures: the degree of a Boolean function as a polynomial over the reals and the block sensitivity. We show that the upper bound on the largest possible separation between these two measures can be improved from $ d^2(f) \geq bs(f) $, established by Tal \cite{tal}, to $ d^2(f) \geq (\sqrt{10} - 2)bs(f) $. As a corollary, we show that the similar upper bounds between some other complexity measures are not tight as well, for instance, we can improve the recent sensitivity conjecture result by Huang \cite{huang} $s^4(f) \geq bs(f) $ to $s^4(f) \geq (\sqrt{10} - 2)bs(f)$. Our techniques are based on the paper by Nisan and Szegedy \cite{nisan_szegedy} and include more detailed analysis of a symmetrization polynomial.

In our next result we show the same type of improvement for the separation between the approximate degree of a Boolean function and the block sensitivity: we show that $\deg_{1/3}^2(f) \geq \sqrt{6/101} bs(f)$ and improve the previous result by Nisan and Szegedy \cite{nisan_szegedy} $ \deg_{1/3}(f) \geq \sqrt{bs(f)/6} $. In addition, we construct an example showing that the gap between the constants in the lower bound and in the known upper bound is less than $0.2$.

In our last result we study the properties of a conjectured fully sensitive function on 10 variables of degree 4, existence of which would lead to improvement of the biggest known gap between these two measures. We prove that there is the only univariate polynomial that can be achieved by symmetrization of such a function by using the combination of interpolation and linear programming techniques.

\keywords{degree of a Boolean function \and approximate degree \and block sensitivity}
\end{abstract}

\section{Introduction}

Let $ f $: $ \{0, 1\}^n \to \{0, 1\} $ be a Boolean function. We can represent $f$ in many ways, for example, as a polynomial over the reals. It is easy to show that every Boolean function can be uniquely represented by such a polynomial (see \cite[exercise 2.23]{jukna}), so we can introduce a complexity measure that is the degree of the polynomial that represents $f$, denoted by $ d(f)$. Another representation of $f$ related to polynomials is the approximating one: a polynomial is called a $ \varepsilon $-approximation of $f$ if for any $ x \in \{0, 1\}^n $ we have $ |f(x) - p(x)| \leq \varepsilon $. Such polynomials make sense for any $ 0 < \varepsilon < \frac{1}{2} $, and it is often assumed that $ \varepsilon = \frac{1}{3} $. By $ \deg_\varepsilon(f) $ we denote the minimum degree among the polynomials that $ \varepsilon $-approximates $f$.

Exact and approximation degrees are closely related to the model called decision trees. The main measure in this model is a decision tree complexity $ D(f) $, which is equal to the amount of bits in input we need to ask in order to give the value of $f$ on such input. Other complexity measures include a sensitivity $s(f)$ and a block sensitivity $bs(f)$. If we denote
\[ x^{(R)} =
\begin{cases}
	1 - x_i & i \in R \\
	x_i & i \notin R
\end{cases} 
\]
then a local block sensitivity $ bs(f, x) $ is the largest amount of disjoint blocks $ R_1, \ldots, R_t $ such that $ f(x) \neq f(x^{(R_i)}) $ for every $ i = 1, \ldots, t $. A block sensitivity in general is the maximum over the local block sensitivities for $ x \in \{0, 1\}^n $. A local sensitivity and sensitivity defined similarly with a restriction that all the blocks must be of size 1. See the \cite{DBLP:journals/tcs/BuhrmanW02} for an overview of these and other complexity measures in the decision tree model.

One of the questions involving various complexity measures is determining the relations between them. For example, recently Huang resolved \cite{huang} the well-known sensitivity conjecture and established that $ s^4(f) \geq bs(f)$. As for polynomials, the first result of this kind was made by Nisan and Szegedy: they analyzed symmetrizations of Boolean functions and showed that $ 2d^2(f) \geq bs(f) $ \cite{nisan_szegedy}. Later, Tal improved this bound by a constant factor by studying a function composition, proving that $ d^2(f) \geq bs(f) $ \cite{tal}. However, the best known example with low degree and high block sensitivity is due to Kushilevitz \cite[Example 6.3.2]{Kushilevitz}, in which $ bs(f) = n = 6^k $ while $ d(f) = 3^k = n^{\log_6{3}} \simeq n^{0.61} $. That means there is still a large gap between the upper and lower bounds in this separation. Our result is the next constant factor improving:

\begin{theorem} \label{main}
    For all Boolean functions $ \{0, 1\}^n \to \{0, 1\} $, we have
    \begin{equation}
        d^2(f) \geq (\sqrt{10} - 2)bs(f) \simeq 1.16bs(f)
    \end{equation}
\end{theorem}

As a corollary of this result, we also improve some other relations between complexity measures, including the Huang's result: we prove that $s^4(f) \geq (\sqrt{10} - 2)bs(f)$.

As for approximating polynomials, Nisan and Szegedy proved that $ \deg_{1/3}(f) \geq \sqrt{bs(f)/6} $ and provided an example (namely, the $ OR_n $ function), for which the constrain is tight up to a constant factor. Later, similar results were archived for this and other Boolean functions \cite{DBLP:journals/corr/abs-0803-4516,DBLP:conf/icalp/BunT13,DBLP:conf/approx/BogdanovMTW19}. In presented papers, authors were not interested in a constant factor in bounds. In our result, we improve the constant in the lower bound and prove that
\begin{theorem} \label{approximate_main}
    For all Boolean functions $ \{0, 1\}^n \to \{0, 1\} $, we have
    \begin{equation}
        \deg_{1/3}^2(f) \geq \sqrt{\frac{6}{101}} bs(f) \simeq 0.24bs(f)
    \end{equation}
\end{theorem}

We also provide an example of a Boolean function (namely, the $ NAE_n$) that can be approximated with a polynomial of degree asymptotically tight to our bound and with a low constant factor in it; in fact, the difference is less than $ 0.2 $, which shows that the lower bound is not far from optimal.

Another way to approach the problem of the separation between $d(f)$ and $ bs(f) $ is to provide examples of functions of low degree and high sensitivity. The first known example was given by Nisan and Szegedy: like in the Kushilevitz's function, $f$ is fully sensitive, depends on $ n = 3^k $ variables and $ d(f) = 2^k = n^{\log_3{2}} \simeq n^{0.63} $. Both examples achieved by composing the base function with itself arbitrary amount of times, and one can show that in the fully sensitive case in such a composition both $ d(f) $ and $ bs(f) $ remain the same in terms of $n$. This technique was later studied by Tal in \cite{tal}. In \cite{nisan_szegedy} the base polynomial consists of 3 variables and has the degree of 2, while in the \cite[Example 6.3.2]{Kushilevitz} it has 6 variables and the degree of 3. In both examples $ 2n = d(d + 1) $, so the next natural step is the fully sensitive $ \tilde{f} $ on 10 variables with $ d(\tilde{f}) = 4 $. Existence of such a function would lead to the new best example of the separation with $ bs(\tilde{f}) = n $ and $ d(\tilde{f}) = n^{\log_{10}{4}} \simeq n^{0.60} $. While we do not provide an example of $ \tilde{f} $, we prove that the only polynomial that can be achieved by symmetrization of it is:
\begin{equation} \label{sym_example}
    \tilde{p}(x) = -\frac{x^{4}}{144}+\frac{5 x^{3}}{36}-\frac{125 x^{2}}{144}+\frac{125 x}{72}
\end{equation}

Our techniques for the lower bounds are based on Nisan and Szegedys' paper. We use the same symmetrization approach but with the more detailed analysis of a symmetrization polynomial: we apply better bounds and study higher order derivatives. As for the upper bounds, we analyze the Chebyshev polynomials of the first kind for approximating polynomials and use the combination of interpolation and linear programming for exact polynomials.

In Section \ref{sect2} we provide necessary definitions and theorems. In Sections \ref{sect3}, \ref{sect4} and \ref{sect5} we prove the lower bound for exact polynomials, the result for approximating polynomials and the property of the low degree function $\tilde{f}$ respectively. 

\section{Preliminaries} \label{sect2}

In this paper we assume that in input of a Boolean function 1 corresponds to the logical true while 0 corresponds to the logical false. The weight of an input is the amount of the positive bits in it. We use the notation $ ||P|| = \sup_{x \in [-1; 1]} |P(x)| $ and denote the set of polynomials of degree at most $d$ by $ \mathcal{P}_d $.

A symmetrization of a polynomial $p$: $\R^n \to \R$ defined as follows:
\begin{equation}
    p^{sym}(x) = \frac{1}{n!} \sum_{\pi \in S_n} p(\pi(x))
\end{equation}

$S_n$ denotes the group of permutations of size $n$ and $ \pi(x) $ denotes the new input, with bits from $x$ moved according to the permutation $\pi$.

The following lemma allows us to represent $ p^{sym} $ as a univariate polynomial of small degree:

\begin{lemma}[Symmetrization lemma \cite{DBLP:books/daglib/0066902}] \label{symmetrization}
    If $p$: $\R^n \to \R$ is a multilinear polynomial, then there exists a univariate polynomial $ \tilde{p} $: $ \R \to \R $ of degree at most the degree of $p$ such that:
    \[ p^{sym}(x_1, \ldots, x_n) = \tilde{p}(x_1 + \ldots + x_n), \quad \forall x \in \{0, 1\}^n \]
\end{lemma}

Note that the value of $ \tilde{p}(k) $ for $ k = 0, 1, \ldots, n $ is equal to the fraction of inputs such that the weight of $x$ is equal to $k$ and $ p(x) = 1 $.

From the proof of this lemma we can also get the explicit formula for the $ \tilde{p} $:
\begin{equation} \label{representation}
    \tilde{p}(x) = c_0 + c_1 \binom{x}{1} + c_2 \binom{x}{2} + \ldots + c_d \binom{x}{d}, \quad d \leq \deg{p}
\end{equation}
By definition for the binomial coefficients. put:
\[ \binom{x}{k} = \frac{x \cdot (x - 1) \cdot \ldots \cdot (x - k + 1)}{k!} \]

The original work of Nisan and Szegedy used the following theorem to bound the degree of a polynomial:

\begin{theorem}[\cite{zeller,rivlin}] \label{previous}
    Let $p$: $ \R \to \R $ be a polynomial such that $ b_1 \leq p(k) \leq b_2 $ for every integer $ 0 \leq k \leq n $ and a derivative satisfies $ |p'(\eta)| \geq c $ for some real $ 0 \leq \eta \leq n $; then
    \[ \deg(p) \geq \sqrt{\dfrac{nc}{c + b_2 - b_1}}. \]
\end{theorem}

However, it is obvious that $ \sqrt{\dfrac{c}{c + b_2 - b_1}} < 1 $, so any bound achieved using this theorem would be weaker than Tal's $ d^2(f) \geq bs(f) $. In order to make progress, we are going to use the following theorem by Ehlich and Zeller, as well as the Markov brothers' inequality:

\begin{theorem}[\cite{zeller}] \label{zeller}
    Let $p$: $ \R \to \R $ be a polynomial of degree $d$. Suppose $n \in \N$ satisfies:
    \begin{enumerate}
        \item $ \rho = \dfrac{d^2(d^2 - 1)}{6n^2} < 1 $, and
        \item $ \forall k = 0, 1, \ldots, n $: $ x_k = -1 + \dfrac{2k}{n} $, $ |p(x_k)| < 1 $
    \end{enumerate}
    then $ ||p|| \leq \dfrac{1}{1 - \rho} $.
\end{theorem}

\begin{theorem}[Markov brothers' inequality] \label{markov}
    For any $ p \in \mathcal{P}_d $ and $ k < d $:
    \begin{equation}
        ||p^{(k)}|| \leq \frac{d^2 \cdot (d^2 - 1) \cdot \ldots \cdot (d^2 - k + 1)}{1 \cdot 3 \cdot \ldots \cdot (2k - 1)}||p||
    \end{equation}
\end{theorem}

A Boolean function $f$ is called fully sensitive at 0 iff $ f(0) = 0 $ and $ s(f, 0) = n $. The next theorem by Nisan and Szegedy explains why it is enough for us to focus only on fully sensitive functions.

\begin{theorem}[\cite{nisan_szegedy}] \label{reduction}
    For every Boolean function $f$ there exists the fully sensitive at 0 function $ \tilde{f} $ that depends on $ bs(f) $ variables and $ d(\tilde{f}) \leq d(f) $.  
\end{theorem}

The simple proofs of Theorems \ref{symmetrization}, \ref{zeller} and \ref{reduction} are given in Appendix \ref{A}. While the proof of the Markov brothers' inequality is significantly harder, two “book-proof”s of it are given in \cite{markov}. 

\section{Exact Polynomials} \label{sect3}

In this section we study the separation between the degree of a Boolean function as an exact polynomial and the block sensitivity. The result organized as follows. Firstly, we will prove the warm-up result, in which we introduce a new approach for bounding degree of a polynomial. Secondly, we will prove a series of lemmas for the main result and prove Theorem \ref{main}.

\subsection{Warm-up}

In order to show new techniques, first we are going to prove a simpler result, which, however, is still better than the previously known upper bound for this separation:

\begin{theorem} \label{warm-up}
    For all Boolean functions $ \{0, 1\}^n \to \{0, 1\} $, we have
    \begin{equation}
        d^2(f) \geq \sqrt{6/5}bs(f) \simeq 1.09bs(f)
    \end{equation}
\end{theorem}

First of all, we need to derive a new approach to bound the degree of a polynomial that would be stronger than Theorem \ref{previous}.

\begin{theorem} \label{approach}
    Let $p$: $ \R \to \R $ be a polynomial of degree $d$ such that:
    
    \begin{enumerate}
        \item $ \forall i = 0, 1, \ldots, n $: $ 0 \leq p(i) \leq 1 $, and
        \item $ \sup_{x \in [0; n]} |p^{(k)}(x)| \geq c $.
    \end{enumerate}
    
    Then either $ d^2 \geq \sqrt{6}n $ or the ratio $ x = \dfrac{d^2}{n} $ satisfies the following inequality:
    \begin{equation} \label{derivative_bound}
        \left( 1 - \frac{x^2}{6} \right) \frac{(2k - 1)!!}{2^{k - 1}}c < x^k
    \end{equation}
    
    $ (2k - 1)!! $ denotes the double factorial: $ (2k - 1)!! = (2k - 1) \cdot (2k - 3) \cdot \ldots \cdot 3 \cdot 1 $.
    
\end{theorem}

\begin{proof}
    Suppose $ d^2 < \sqrt{6}n $, otherwise the first statement holds. In terms of Theorem \ref{zeller}:
    \begin{equation} \label{rho}
        \rho = \frac{d^2(d^2 - 1)}{6n^2} < \frac{d^4}{6n^2} < 1
    \end{equation}
    
    Let $ P(x) $ be defined as follows:
    \[ P(x) = p\left( \frac{n}{2}(x + 1)\right) - \frac{1}{2} \]
    
    By Theorem \ref{zeller}, $ ||P|| \leq \dfrac{1}{2} \cdot \dfrac{1}{1 - \rho} $. On the other hand, $ ||P^{(k)}|| \geq \dfrac{n^k}{2^k} \cdot c $. Combined with Inequality \ref{markov}, we get
    \[ c\left( \frac{n}{2} \right)^k \leq ||P^{(k)}|| \leq \frac{d^2 \cdot (d^2 - 1) \cdot \ldots \cdot (d^2 - k + 1)}{1 \cdot 3 \cdot \ldots \cdot (2k - 1)}||P|| < \frac{d^{2k}}{(2k - 1)!!} \cdot \frac{1}{2(1 - \rho)} \]
    \begin{equation} \label{derivative_proof}
        (1 - \rho) \frac{(2k - 1)!!}{2^{k - 1}}c \leq \left( \frac{d^2}{n} \right)^k
    \end{equation}
    
    By substituting \eqref{rho} and $ x = \dfrac{d^2}{n} $ into \eqref{derivative_proof}, we obtain exactly the inequality \eqref{derivative_bound}.
    
\end{proof}

The same approach was used by Beigel \cite[lemma 3.2]{beigel}, however, he didn't parameterized his result and proved it only for the first derivative. If we apply his result with a trivial bound $\sup_{x \in [0; n]} |p'(x)| \geq 1$, it follows that $ 1 - \dfrac{x^2}{6} < x $ for $ x = \dfrac{d^2}{n} $. As $ x > 0 $, the solution is $ x \geq \sqrt{15} - 3 \simeq 0.87 $, which is stronger than the original bound, but weaker than the Tal's result.

Now we are ready to prove the warm-up result.

\begin{proof}[Theorem \ref{warm-up}]
    Because of reduction \ref{reduction}, we can assume without loss of generality that $f$ is fully sensitive at 0, and so a polynomial $ p $ derived from Lemma $\ref{symmetrization}$ satisfies $ p(0) = 0 $ and $ p(1) = 1 $. Also, it is obvious that if $ \rho $ from Theorem \ref{zeller} is not less than $1$, then $ d^2 \geq \sqrt{6}n > \sqrt{6/5}n $, so we assume that $ \rho < 1 $. 

    Suppose $ p \in \mathcal{P}_2 $, i.e. $ p(x) = ax^2 + bx + c $. From the values of $ p(0), \ p(1) $ and $ p(2) $ we obtain the following constrains:
    \begin{enumerate}
        \item $ p(0) = c = 0 \Ra c = 0 $.
        \item $ p(1) = a + b + c = 1 \Ra a + b = 1 $.
        \item $ p(2) = 4a + 2b \Ra -2 \leq 2a \leq -1 $.
    \end{enumerate}
    As a result, $ |p''(x)| = |2a| \geq 1 $.
    
    In general case, let $ q(x) $ be the quadratic polynomial that equals to $ p(x) $ at $ x \in \{ 0,1,2 \} $, and $ \tilde{p}(x) = p(x) - q(x) $. From our definition it follows that $ \tilde{p}(0) = \tilde{p}(1) = \tilde{p}(2) = 0 $, so there exists $ \xi \in [0; 2] $: $ \tilde{p}''(\xi) = 0 $. Then $ |p''(\xi)| = |q''(\xi) + \tilde{p}''(\xi)| \geq 1 $, and as a direct consequence $ \sup_{x \in [0; n]} |p''(\xi)| \geq 1 $. Substituting $ c = 1 $ and $ k = 2 $ in \eqref{derivative_bound}, we obtain the following inequality for $ x = \dfrac{d^2}{n} $:
    \begin{equation}
        \left( 1 - \frac{x^2}{6} \right) \frac{3}{2} < x^2
    \end{equation}
    
    Since $ x \geq 0 $, we get $ x \geq \sqrt{6/5} $ and $ d^2 \geq \sqrt{6/5}n $.
    
\end{proof}

\subsection{Proof of Theorem \ref{main}}

To increase the constant factor from $ \sqrt{6/5} $ to $ \sqrt{10} - 2 $, we should analyze higher order derivatives. In order to do so, we are going to use representation \eqref{representation}. We also need to establish a series of lemmas. 

\begin{lemma} \label{c3_bound}
    Suppose $f$: $ \{0, 1\}^n \to \{0, 1\} $ is fully sensitive at 0 and $ n \geq 4 $; then for symmetrization polynomial $p$ we have $ \sup_{x \in [0; n]} |p'''(x)| \geq 1 - p(3) $.
\end{lemma}

\begin{lemma} \label{sum_bound}
	\begin{equation}
	    \sum_{k = 4}^{\infty} \frac{1}{k2^{k - 2}} < \frac{1}{8}
	\end{equation}
\end{lemma}

The proofs are omitted to Appendix \ref{A}. With this lemmas we are ready for the main proof.

\begin{proof}[Theorem \ref{main}]
    If $ n < 4 $, then the theorem follows from Tal's bound $ d^2(f) \geq bs(f) $. As in Theorem \ref{warm-up}, we assume that $ f $ is fully sensitive at $0$ and $ \rho < 1 $.

    It is easy to show that
    \[ \binom{x}{k}' \bigg|_{x = 0} = \frac{(-1)^{k + 1}}{k}, \quad k \in \N \]
    
    If we combine this with representation \eqref{representation}, we get the formula for the first derivative of the symmetrization polynomial at $x = 0$:
    \begin{equation} \label{derivative}
        p'(0) = \sum_{k = 1}^d (-1)^{k + 1} \frac{c_k}{k}
    \end{equation}
    
    We can bound the first three coefficients:
    
    \begin{enumerate}
        \item $ p(1) = c_1 = 1 \Ra c_1 = 1 $.
        \item $ p(2) = 2c_1 + c_2 \leq 1 \Ra c_2 \leq -1 $.
        \item $ p(3) = 3c_1 + 3c_2 + c_3 \geq 0 \Ra c_3 \geq p(3) $.
    \end{enumerate}
    
    If $ p(3) < \dfrac{3}{8} $, then by Lemma \ref{c3_bound} $ \sup_{x \in [0; n]} |p'''(x)| > \dfrac{5}{8} $. Substituting $ c = \dfrac{5}{8} $ and $ k = 3 $ in \eqref{derivative_bound}, we get the inequality for $ x = \dfrac{d^2}{n} $:
    
    \[ \left( 1 - \frac{x^2}{6} \right) \frac{75}{32} < x^3 \]
    
    Inequality implies that $ x > 1.2 $, which satisfies the statement of the theorem. In remaining case, we have $ c_3 \geq \dfrac{3}{8} $. Substituting all the constrains in \eqref{derivative}, we obtain:
    \begin{equation}
        p'(0) \geq \frac{3}{2} + \frac{1}{8} + \sum_{k = 4}^d (-1)^{k + 1} \frac{c_k}{k}
    \end{equation}
    
    Suppose $ |c_k| < \dfrac{1}{2^{k - 2}} $ for $k > 3$, then
    \[ p'(0) \geq \frac{3}{2} + \frac{1}{8} - \sum_{k = 4}^d \frac{1}{k2^{k - 2}} > \frac{3}{2} + \frac{1}{8} - \sum_{k = 4}^\infty \frac{1}{k2^{k - 2}} > \frac{3}{2} + \frac{1}{8} - \frac{1}{8} = \frac{3}{2} \]
    
    Inequality \eqref{derivative_bound} with $ c = \dfrac{3}{2} $ and $ k = 1 $ implies $ \left(1 - \dfrac{x^2}{6}\right) \dfrac{3}{2} \leq x $ for $ x = \dfrac{d^2}{n} $. The solution is $ x \geq (\sqrt{10} - 2) $, which means that $ d^2 \geq (\sqrt{10} - 2)n $.
    
    The only case left to consider is if there exists such $k > 3$ that $ |c_k| \geq \dfrac{1}{2^{k - 2}} $. To deal with it, we first need to show that $ \sup_{x \in [0; n]} |p^{(k)}(x)| \geq c_k $. If $ d = k $, then the derivative is a constant and equals to $c_k$ because $ \binom{x}{k}^{(k)} = 1 $. In other case, let $ q(x) $ be the polynomial that consists of all the terms from \eqref{representation} up to one with the $ c_k $. Then $ \tilde{p}(x) = p(x) - q(x) $ equals to 0 for $ x = 0, 1, \ldots, k $, so $ \exists \xi \in [0; k]$: $ \tilde{p}^{(k)}(\xi) = 0 $ and $ |p^{(k)}(\xi)| = |c_k| $.
    
    Now, for $ k = 4, 5 $ we obtain the following inequalities from \eqref{derivative_bound}:
    \[ \left(1 - \frac{x^2}{6}\right) \frac{105}{32} \leq x^4 \quad \Ra \quad x > 1.24 > (\sqrt{10} - 2) \]
	\[ \left(1 - \frac{x^2}{6}\right) \frac{945}{128} \leq x^5 \quad \Ra \quad x > 1.3 > (\sqrt{10} - 2) \]
	
	For $k > 5$, we need to show that the solution from the inequality would not be worse than for $ k = 5 $. Notice that every time we increase $k$ in inequality \eqref{derivative_bound} we multiply the left hand side by $ \dfrac{2k + 1}{4} > \sqrt{6} $ and the right hand side by $x$. But if we recall that $ \rho < 1 $, we get $ x < \sqrt{6} $, and thus the inequality becomes tighter. As a result, the statement holds for $k > 5$ as well.
    
\end{proof}

\subsection{Corollaries}

While our improvement may seem insignificant, it shows that the currently known bound between $ d(f) $ and $ bs(f) $ is not tight. Next two corollaries shows that the same holds for some other pairs of complexity measures.

\begin{corollary}
    For all Boolean functions $ \{0, 1\}^n \to \{0, 1\} $, we have
    \begin{equation}
        s^4(f) \geq (\sqrt{10} - 2)bs(f)
    \end{equation}
\end{corollary}

\begin{proof}
    In the proof of the sensitivity conjecture \cite{huang}, Huang established the following bound for $d(f)$:
    \begin{equation}
        s^2(f) \geq d(f)
    \end{equation}
    
    Combined with Theorem \ref{main}, it follows that
    \[ s^4(f) \geq d^2(f) \geq (\sqrt{10} - 2)bs(f) \]
    
\end{proof}

\begin{corollary}
    For all Boolean functions $ \{0, 1\}^n \to \{0, 1\} $, we have
    \begin{equation}
        d^3(f) \geq (\sqrt{10} - 2)D(f)
    \end{equation}
\end{corollary}

\begin{proof}
    Combining Theorem \ref{main} with the bound $ D(f) \leq bs(f) \cdot d(f) $ from the paper \cite{midrijanis2004exact}, we get
    \[ d^3(f) \geq d(f) \cdot (\sqrt{10} - 2)bs(f) \geq (\sqrt{10} - 2)D(f) \]
    
\end{proof}

\section{Approximating Polynomials} \label{sect4}

In this section, we improve the constant factor in the separation between the degree of an approximating polynomial and the block sensitivity and provide an example of a polynomial for $ NAE_n $ function that shows that not only our bound is asymptotically tight, but the difference between the best known constant in the lower bound and the constant in our example is relatively small as well.

\subsection{Lower bound}

Before the proof we need to derive a similar to \ref{approach} lemma, but this time for approximating polynomials.

\begin{lemma}
    Let $p$: $ \R \to \R $ be a polynomial of degree $d$ such that:
    
    \begin{enumerate}
        \item $ \forall i = 0, 1, \ldots, n $: $ -\dfrac{1}{3} \leq p(i) \leq \dfrac{4}{3} $, and
        \item $ \sup_{x \in [0; n]} |p^{(k)}(x)| \geq c $.
    \end{enumerate}
    
    Then either $ d^2 \geq \sqrt{6}n $ or the ratio $ x = \dfrac{d^2}{n} $ satisfies the following inequality:
    \begin{equation} \label{approxiamte_derivative_bound}
        \left( 1 - \frac{x^2}{6} \right) \frac{(2k - 1)!!}{2^k} \cdot \frac{6c}{5} < x^k
    \end{equation}
    
\end{lemma}

\begin{proof}
    The only difference between this lemma and Theorem \ref{approach} is the bounds for $ p(k) $, therefore if we define $ P(x) $ the same as earlier, we get the weaker upper bound: $ ||P|| \leq \dfrac{5}{6} \cdot \dfrac{1}{1 - \rho} $. The remaining part of the proof is the same as in \ref{approach}. 
\end{proof}

\begin{proof}[Theorem \ref{approximate_main}]
    Using reduction \ref{reduction}, we can assume that the symmetrization polynomial of $f$ satisfies $ -\dfrac{1}{3} \leq p(0) \leq \dfrac{1}{3} $ and $ \dfrac{2}{3} \leq p(1) \leq \dfrac{4}{3} $. As always, we can only consider the case $ \rho < 1 $. Also, if $ n < 5 $, then the theorem follows from the original bound $ 6\deg^2_{1/3}(f) \geq bs(f)$.
    
    Suppose that $ p \in \mathcal{P}_3 $. By Lagrange's interpolation formula for $ x \in \{0, 1, 2\} $ and $ x \in \{ 0, 2, 5 \} $:
    \begin{equation} \label{app1}
        p(x) = \frac{(x - 1)(x - 2)}{2}p(0) + x(2 - x)p(1) + \frac{x(x - 1)}{2}p(2)
    \end{equation} 
    \begin{equation} \label{app2}
        p(x) = \frac{(x - 2)(x - 5)}{10}p(0) - \frac{x(x - 5)}{6}p(2) + \frac{x(x - 2)}{15}p(5)
    \end{equation}
    
    From \eqref{app1} we get $ \forall x \quad p''(x) = p(0) - 2p(1) + p(2) \leq -1 + p(2) $, and if $ p(2) \leq \dfrac{14}{15} $, then $ p''(x) \leq -\dfrac{1}{15} $. Otherwise, from \eqref{app2} we get $ \forall x \quad p(x) = \dfrac{1}{5}p(0) - \dfrac{1}{3}p(2) + \dfrac{2}{15}p(5) \leq \dfrac{11}{45} - \dfrac{1}{3}p(2) \leq -\dfrac{1}{15}$. If $ \deg{p} > 3 $, we can use the same reduction as in the proof of Theorem \ref{warm-up}.
    
    Applying \eqref{approxiamte_derivative_bound} with $ c = \dfrac{1}{15} $ and $ k = 2 $, we obtain the following inequality:
    
    \begin{equation}
        \left(1 - \frac{x^2}{6}\right) \frac{3}{50} < x^2
    \end{equation}
    
    It now follows that $ x \geq \sqrt{\dfrac{6}{101}} $ and $\deg_{1/3}^2(f) \geq \sqrt{\dfrac{6}{101}} bs(f)$.
    
\end{proof}

\subsection{Upper bound}

A function $ NAE_n $: $ \{0, 1\}^n \to \{0, 1\} $ equals to 1 iff $ x \in \{0^n, 1^n\} $, i.e. all the bits in the input are the same. The next theorem provides a polynomial that approximates $ NAE_n $ and gives the upper bound for $ \deg_{1/3}(f) $ in terms of the block sensitivity.

\begin{theorem} \label{example}
    Define $ d = \lceil\sqrt{c(n - 2)}\rceil $ with a constant $ c $ satisfying the following inequality:
    \begin{equation} \label{constant_inequality}
        2c + \frac{2}{3}c^2 - \frac{2c}{3(n - 2)} > 1
    \end{equation}
    
    Then there exists a polynomial of degree $ d $ if $d$ is even and $ d + 1 $ otherwise that is a $ \frac{1}{3} $-approximation of $ NAE_n $. 
\end{theorem}

\begin{proof}
    In our construction, we use the Chebyshev polynomials of the first kind, defined as $ T_k(x) = \cos{(k \arccos{x})} $. We need the following properties of them; proof of property 3 is omitted to Appendix \ref{A}, and property 4 is \cite[lemma 5.17]{markov} for $ k = 1 $ and $ k = 2 $.
    
    \begin{enumerate}
        \item If $k$ is even, then $ T_k(x) = T_k(-x)$. 
        \item $ \forall x \in [-1; 1] \quad |T_k(x)| \leq 1$.
        \item $ T_k''(\theta) \geq T_k''(1) $ for $ \theta \geq 1 $.
        \item $ T_k'(1) = k^2 $ and $ T_k''(1) = \dfrac{k^4 - k^2}{3} $.
    \end{enumerate}
    
    By definition, put
    \begin{equation}
        p(x) = 1 - \dfrac{2T_k(\frac{2x - n}{n - 2})}{3T_k(\frac{n}{n - 2})}
    \end{equation}
    
    It is clear from property 1 that $ p(0) = p(n) = \dfrac{1}{3} $. If we show that $ T_k\left(\dfrac{n}{n - 2}\right) \geq 2 $, then by property 2 for all $ 1 \leq k \leq n - 1 $ we have
    \[ \left| \frac{2T_k(\frac{2x - n}{n - 2})}{3T_k(\frac{n}{n - 2})} \right| \leq \frac{2 \cdot 1}{3 \cdot 2} = \frac{1}{3} \quad \Ra \quad \frac{2}{3} \leq p(k) \leq \frac{4}{3} \]
    
    and $ q(x_1, \ldots, x_n) = p(n - x_1 - \ldots - x_n) $ is indeed the $ \frac{1}{3} $-approximation of $ NAE_n $. 
    
    Substituting $ x = 1 $ in the Taylor series for $ T_k(\dfrac{n}{n - 2}) $, we obtain
    \begin{equation}\begin{array}{c}
        T_k\left( \dfrac{n}{n - 2} \right) = T_k(1) + \dfrac{2}{n - 2}T_k'(1) + \dfrac{2}{(n - 2)^2}T_k''(\theta) \geq \\
        \\
        \geq T_k(1) + \dfrac{2}{n - 2}T_k'(1) + \dfrac{2}{(n - 2)^2}T_k''(1)
    \end{array}\end{equation}
    
    The last inequality holds because of property 3. Combining property 4 and \eqref{constant_inequality}, we get:
    \begin{equation}\begin{array}{c}
        T_k\left( \dfrac{n}{n - 2} \right) \geq 1 + 2c + \dfrac{2}{3(n - 2)^2}(c^2(n - 2)^2 - c(n - 2)) =
        \\
        \\
        = 1 + 2c + \dfrac{2}{3}c^2 - \dfrac{2c}{3(n - 2)} \geq 2
    \end{array}\end{equation}
    
\end{proof}

Because the last term in the left hand side of \eqref{constant_inequality} tends to zero as $n$ tends to infinity, the optimal $c$ tends to the solution of the following inequality: $ 2x + \dfrac{2}{3}x^2 > 1 $. The solution is $ x > \dfrac{1}{2}(\sqrt{15} - 3) \simeq 0.43 $, so the difference between $ c $ and the best known lower bound is less than $0.2$, which shows that the bound \ref{approximate_main} is close to be tight.

\section{Fully Sensitive Function of Small Degree} \label{sect5}

The last result of this paper is about an example of a function with low degree and high block sensitivity. We study properties of conjectured function $ \tilde{f} $ on $10$ variables with $ d(\tilde{f}) = 4 $. By applying the same composition scheme as in the previous examples, we can generalize $ \tilde{f} $ for the arbitrary large $n$. While we do not provide an example of $\tilde{f}$, we prove that if $\tilde{f}$ is fully sensitive at 0, then by applying Lemma \ref{symmetrization} to it, the only univariate polynomial we can get is \eqref{sym_example}.

We prove this statement in two steps. Firstly, we achieve such a polynomial by interpolation. Secondly, we prove the uniqueness using the linear programming.

\subsection{Interpolation}

The first part of the proof is to construct polynomial \eqref{sym_example}. We do this by establishing the extremal property of all the symmetrizations of degree 4 for $n \geq 8$.

\begin{theorem}
    Let $f$: $ \{0, 1\}^n \to \{0, 1\} $ be fully sensitive at 0 and $ n \geq 8 $; then for symmetrization polynomial $p$ we have $ \sup_{x \in [0; n]} |p^{(4)}(x)| \geq \frac{1}{6} $. Moreover, the only polynomial for which inequality is tight is \eqref{sym_example}.
\end{theorem}

\begin{proof}
    Using the Lagrange’s interpolation formula for $ x \in \{0,1,2,7,8\} $ and $ x \in \{0,1,2,5,7\} $, we get the following representations:
    
    \begin{enumerate} 
        \item \begin{equation}\label{p1}\begin{array}{c}
            \forall x \quad p^{(4)}(x) = -\dfrac{24}{1 \cdot 1 \cdot 6 \cdot 7} + \dfrac{24}{2 \cdot 1 \cdot 5 \cdot 6} p(2) - \dfrac{24}{7 \cdot 6 \cdot 5 \cdot 1} p(7) + \\
            \\
            + \dfrac{24}{8 \cdot 7 \cdot 6 \cdot 1} p(8) = -\dfrac{4}{7} + \dfrac{2}{5} p(2) - \dfrac{4}{35} p(7) + \dfrac{1}{14} p(8) \leq -\dfrac{1}{10} - \dfrac{4}{35} p(7)
        \end{array}\end{equation}
        \item \begin{equation}\label{p2}\begin{array}{c}
            \forall x \quad p^{(4)}(x) = -\dfrac{24}{1 \cdot 1 \cdot 4 \cdot 6} + \dfrac{24}{2 \cdot 1 \cdot 3 \cdot 5} p(2) - \dfrac{24}{5 \cdot 4 \cdot 3 \cdot 2} p(5) + \\
            \\
            + \dfrac{24}{7 \cdot 6 \cdot 5 \cdot 2} p(7) = -1 + \dfrac{4}{5} p(2) - \dfrac{1}{5} p(5) + \dfrac{2}{35} p(7) \leq -\dfrac{1}{5} + \dfrac{2}{35} p(7)
        \end{array}\end{equation}
    \end{enumerate}
    
    If $ p(7) \leq \dfrac{7}{12} $, then from \eqref{p1} we get $ p^{(4)}(x) \leq -\dfrac{1}{6} $, otherwise we get the same result from \eqref{p2}. Inequality is tight iff $ p(7) = \dfrac{7}{12}, \ p(2) = 1, \ p(5) = 0 $ and $ p(8) = 1 $. Combined with $ p(0) = 0 $ and $ p(1) = 1 $, we obtain that there is the only polynomial of degree at most 5 that satisfies all the constrains. By applying the Lagrange’s interpolation formula, we get polynomial \eqref{sym_example}. 
\end{proof}

Note that \eqref{sym_example} is indeed a symmetrization polynomial for some function, because for $ k = 0, \ldots, 10 $ the values $ p(k) $ represent the fraction of inputs for corresponding weights (i.e. $\binom{n}{k} \cdot p(k) \in \N_0$).

\subsection{Linear programming}

The second part of the proof is to show that \eqref{sym_example} is the only symmetrization polynomial for $ n = 10 $ and $ d(f) \leq 4 $. This part is done with a linear programming solver. In our case we are going to use the scipy.optimize.linprog, the full code for the problem is available at Google Colab \cite{colab}.

\begin{theorem}
    The only symmetrization polynomial for the fully sensitive at 0 function of 10 variables with degree at most 4 is \eqref{sym_example}.
\end{theorem}

\begin{proof}
    Suppose $ p(x) = c_1x + c_2x^2 + c_3x^3 + c_4x^4 $ is the needed polynomial. We use the necessary (but not sufficient) conditions for $p(x)$ to create a linear programming task. Namely, we require $ p(1) = 1, \ \forall k = 2, \ldots, 10$: $ 0 \leq p(k) \leq 1 $ and some additional constrain for $c_4$. We require a solver to minimize $c_4$.
    
    If we add a constrain $ c_4 > 0 $, solver proves that problem is infeasible. Without any constrains for $ c_4 $, solver states that the solution is $ -\dfrac{1}{144} $, i.e. the minimum value for $ p^{(4)}(x) $ for which solver can find the solution is $ -\dfrac{1}{144} \cdot 24 = -\dfrac{1}{6} $. But we know that the only polynomial for this value is \eqref{sym_example}, thus we obtain it's uniqueness.
\end{proof}

\section{Conclusion}

Although we made improvements in relations between some complexity measures, we strongly suspect that the current results are still not tight. For example, the choice of points for interpolation in many theorems was not really justified, so we think that finding a pattern for the choice of interpolation set is one of the keys for the further improvements. Also, we suspect that by using Bernstein's inequality (see \cite[theorem 5.1.7]{borwein1995polynomials}) in Theorem \ref{main} with or instead of the Markov brothers' inequality for the first derivative one might improve the result as well.

Another open question occurs if we add further restrictions for the function $f$, for example, if we want $f$ to be symmetrical. It was proved by von zur Gathen and Roche that if $ n = p - 1 $ for prime $p$, then $ d(f) = n $, and as a corollary in general $ d(f) = n - \mathcal{O}(n^{0.525}) $ \cite{DBLP:journals/combinatorica/GathenR97}. It is conjectured that $ d(f) \geq n - 3 $, but very little progress was made since. For instance, the result by Cohen and Shpilka \cite{DBLP:journals/eccc/CohenS10} states that if $ n = p^2 - 1 $, then $ d(f) \geq n - \sqrt{n} $.

\subsubsection{Acknowledgments.} Author would like to thank Vladimir V. Podolskii for the proof idea for Theorem \ref{warm-up}.

\bibliographystyle{splncs04}
\bibliography{sample}

\begin{thebibliography}{10}
\providecommand{\url}[1]{\texttt{#1}}
\providecommand{\urlprefix}{URL }
\providecommand{\doi}[1]{https://doi.org/#1}

\bibitem{beigel}
Beigel, R.: Perceptrons, pp, and the polynomial hierarchy. Comput. Complex.
  \textbf{4},  339--349 (1994). \doi{10.1007/BF01263422},
  \url{https://doi.org/10.1007/BF01263422}

\bibitem{DBLP:conf/approx/BogdanovMTW19}
Bogdanov, A., Mande, N.S., Thaler, J., Williamson, C.: Approximate degree,
  secret sharing, and concentration phenomena. In: Achlioptas, D., V{\'{e}}gh,
  L.A. (eds.) Approximation, Randomization, and Combinatorial Optimization.
  Algorithms and Techniques, {APPROX/RANDOM} 2019, September 20-22, 2019,
  Massachusetts Institute of Technology, Cambridge, MA, {USA}. LIPIcs,
  vol.~145, pp. 71:1--71:21. Schloss Dagstuhl - Leibniz-Zentrum f{\"{u}}r
  Informatik (2019). \doi{10.4230/LIPIcs.APPROX-RANDOM.2019.71},
  \url{https://doi.org/10.4230/LIPIcs.APPROX-RANDOM.2019.71}

\bibitem{borwein1995polynomials}
Borwein, P., Erdelyi, T.: Polynomials and Polynomial Inequalities. Graduate
  Texts in Mathematics, Springer New York (1995),
  \url{https://books.google.ru/books?id=386CC7JnuuwC}

\bibitem{DBLP:journals/tcs/BuhrmanW02}
Buhrman, H., de~Wolf, R.: Complexity measures and decision tree complexity: a
  survey. Theor. Comput. Sci.  \textbf{288}(1),  21--43 (2002).
  \doi{10.1016/S0304-3975(01)00144-X},
  \url{https://doi.org/10.1016/S0304-3975(01)00144-X}

\bibitem{DBLP:conf/icalp/BunT13}
Bun, M., Thaler, J.: Dual lower bounds for approximate degree and
  markov-bernstein inequalities. In: Fomin, F.V., Freivalds, R., Kwiatkowska,
  M.Z., Peleg, D. (eds.) Automata, Languages, and Programming - 40th
  International Colloquium, {ICALP} 2013, Riga, Latvia, July 8-12, 2013,
  Proceedings, Part {I}. Lecture Notes in Computer Science, vol.~7965, pp.
  303--314. Springer (2013). \doi{10.1007/978-3-642-39206-1\_26},
  \url{https://doi.org/10.1007/978-3-642-39206-1\_26}

\bibitem{DBLP:journals/eccc/CohenS10}
Cohen, G., Shpilka, A.: On the degree of symmetric functions on the boolean
  cube. Electron. Colloquium Comput. Complex.  \textbf{17}, ~39 (2010),
  \url{http://eccc.hpi-web.de/report/2010/039}

\bibitem{zeller}
Ehlich, H., Zeller, K.: Schwankung von polynomen zwischen gitterpunkten.
  Mathematische Zeitschrift pp. 41--44 (1964)

\bibitem{DBLP:journals/combinatorica/GathenR97}
von~zur Gathen, J., Roche, J.R.: Polynomials with two values. Comb.
  \textbf{17}(3),  345--362 (1997). \doi{10.1007/BF01215917},
  \url{https://doi.org/10.1007/BF01215917}

\bibitem{Kushilevitz}
Hatami, P., Kulkarni, R., Pankratov, D.: Variations on the sensitivity
  conjecture. Theory Comput.  \textbf{4},  1--27 (2011).
  \doi{10.4086/toc.gs.2011.004}, \url{https://doi.org/10.4086/toc.gs.2011.004}

\bibitem{huang}
Huang, H.: Induced subgraphs of hypercubes and a proof of the sensitivity
  conjecture. Annals of Mathematics  \textbf{190}(3),  949--955 (2019),
  \url{https://www.jstor.org/stable/10.4007/annals.2019.190.3.6}

\bibitem{jukna}
Jukna, S.: Boolean Function Complexity - Advances and Frontiers, Algorithms and
  combinatorics, vol.~27. Springer (2012). \doi{10.1007/978-3-642-24508-4},
  \url{https://doi.org/10.1007/978-3-642-24508-4}

\bibitem{midrijanis2004exact}
Midrijanis, G.: Exact quantum query complexity for total boolean functions.
  arXiv preprint quant-ph/0403168  (2004)

\bibitem{DBLP:books/daglib/0066902}
Minsky, M., Papert, S.: Perceptrons - an introduction to computational
  geometry. {MIT} Press (1987)

\bibitem{nisan_szegedy}
Nisan, N., Szegedy, M.: On the degree of boolean functions as real polynomials.
  Comput. Complex.  \textbf{4},  301--313 (1994). \doi{10.1007/BF01263419},
  \url{https://doi.org/10.1007/BF01263419}

\bibitem{colab}
Proskurin, N.: Symmetrization linprog.
  \url{https://colab.research.google.com/drive/1XKJSYLElVxGgZuwHaFy4BdoTN4JIgKXJ?usp=sharing}
  (2020)

\bibitem{rivlin}
Rivlin, T.J., Cheney, E.W.: A comparison of uniform approximations on an
  interval and a finite subset thereof. SIAM Journal on Numerical Analysis
  \textbf{3}(2),  311--320 (1966), \url{http://www.jstor.org/stable/2949624}

\bibitem{markov}
Shadrin, A.: Twelve proofs of the markov inequality. Approximation theory: a
  volume dedicated to Borislav Bojanov pp. 233--298 (2004)

\bibitem{DBLP:journals/corr/abs-0803-4516}
Spalek, R.: A dual polynomial for {OR}. CoRR  \textbf{abs/0803.4516} (2008),
  \url{http://arxiv.org/abs/0803.4516}

\bibitem{tal}
Tal, A.: Properties and applications of boolean function composition. In:
  Kleinberg, R.D. (ed.) Innovations in Theoretical Computer Science, {ITCS}
  '13, Berkeley, CA, USA, January 9-12, 2013. pp. 441--454. {ACM} (2013).
  \doi{10.1145/2422436.2422485}, \url{https://doi.org/10.1145/2422436.2422485}

\end{thebibliography}

\appendix

\section{Omitted Proofs} \label{A}

\begin{proof}[Lemma \ref{symmetrization}]
    Define $ d = \deg{p^{sym}} $ and $ P_k = \sum_{|S| = k} \prod_{i \in S} x_i $ where $S$ are chosen from the subsets of $ [n] = \{1, 2, \ldots, n\} $. Suppose that $S$ is a monomial in $ p $; then by definition symmetrization adds up all the monomials of size $|S|$ to $ p^{sym} $ equal amount of times: in order to get a specific monomial $ S' $ one should fix the permutation of variables in $ S' $ and in $ [n] \setminus S' $, thus amount of every monomial is equal. Therefore, we can rewrite $ p^{sym} $ as
    \[ p^{sym}(x) = c_0 + c_1 P_1(x) + \ldots + c_d P_d(x) \]
    
    Note that we only interested in $ x \in \{0, 1\}^n $; therefore, every term in $ P_k $ is equal to 1 iff every variable in it is equal to 1. Thus it is obvious that if $ z = x_1 + x_2 + \ldots + x_n $, then 
    \[ p^{sym}(x) = c_0 + c_1 \binom{z}{1} + \ldots + c_d \binom{z}{d} = \tilde{p}(z) \]
    
    $ \deg{\tilde{p}} \leq \deg{p} $ because $ \deg{p^{sym}} \leq \deg{p} $.
\end{proof}

\begin{proof}[Theorem \ref{zeller}]
    Define $ K = \inf\{ k : ||p|| \leq 1 + k \} $. From Inequality \ref{markov} we get
    \begin{equation} \label{second_derivative_bound}
        ||p''(x)|| \leq \frac{d^2(d^2 - 1)}{3}(1 + K)
    \end{equation}
    
    Let $ \xi $ be the point of maximum on $ [-1; 1] $, i.e. $ ||p|| = |p(\xi)| $. The cases $ \xi = \pm 1 $ are trivial because $ |p(\xi)| \leq 1 $, so we can assume that $ \xi $ is an inner point and $ p'(\xi) = 0 $. Also, because $ \forall k = 0, 1, \ldots, n - 1 \quad x_{k + 1} - x_k = \dfrac{2}{n} $ there exists such $k$ that $ \left| x_k - \xi \right| \leq \dfrac{1}{n} $. Applying Taylor series for $ p(x) $, we obtain
    \[ p(x_k) = p(\xi) + (x_k - \xi)p'(\xi) + \frac{(x_k - \xi)^2}{2}p''(\theta) = p(\xi) + \frac{(x_k - \xi)^2}{2}p''(\theta), \quad \theta \in [-1;1] \]
    
    Substituting \eqref{second_derivative_bound} in the last equality, we get another bound for $||p||$:
    \begin{equation}\begin{array}{c}
        ||p|| = |p(\xi)| = \left|p(x_k) - \dfrac{(x_k - \xi)^2}{2}p''(\theta)\right| \leq |p(x_k)| + \left|\dfrac{(x_k - \xi)^2}{2}p''(\theta)\right| \leq \\
        \\
        1 + \dfrac{1}{2n^2} \cdot \dfrac{d^2(d^2 - 1)}{3}(1 + K) =  1 + \rho(1 + K)
    \end{array}\end{equation}
    
    By definition $ K \leq \rho(1 + K) $ and as a corollary $K \leq \dfrac{\rho}{1 - \rho}$ and $ ||p|| \leq 1 + K \leq \dfrac{1}{1 - \rho} $.
    
\end{proof}

\begin{proof}[Theorem \ref{reduction}]
    Let $ x $ be the input such that $ bs(f) = bs(f, x) $, and let $ S_1, S_2, \ldots, S_t $ be the blocks on which we achieve such block sensitivity. Without loss of generality we can assume that $ f(0) = 0 $, otherwise we can introduce the new function $ g(x) = 1 - f(x) $.
    
    Let $\tilde{f}$ be defined as follows
    \begin{equation}
         \tilde{f}(y_1, \ldots, y_t) = f(x \oplus y_1S_1 \oplus \ldots \oplus y_tS_t)
    \end{equation}
    i.e. we create a new input for $f$ such that every bit $x_j$ is equal to $ x_j \oplus y_i $ if $ x_j \in S_i$ or $x_j$ is left unchanged otherwise.
    
    $ d(\tilde{f}) \leq d(f) $ because $ \tilde{f} $ is a linear substitution in $ f $. On the other hand, $ \tilde{f} $ is fully sensitive at 0 as $ \tilde{f}(0) = f(0) = 0 $ and $ f(e_j) = f(x^{(S_i)}) = 1 $. Thus $ \tilde{f} $ satisfies the statement as $ t = bs(f)$.
    
\end{proof}

\begin{proof}[Lemma \ref{c3_bound}]
    Suppose that $ p \in \mathcal{P}_3 $. Using the Lagrange's interpolation formula for $ x \in \{0,1,3,4\} $, we get the following representation:
    \[ p(x) = \frac{x(x - 3)(x - 4)}{6} - \frac{x(x - 1)(x - 4)}{6}p(3) + \frac{x(x - 1)(x - 3)}{12}p(4) \]
	\[ \forall x \quad p'''(x) = 1 - p(3) + \frac{p(4)}{2} \geq 1 - p(3)  \]
	
	In general case, let $ q \in \mathcal{P}_3 $ be equal to $p$ on the same set of points. Similarly to Theorem \ref{warm-up}, if $ \tilde{p}(x) = p(x) - q(x) $, then $ \exists \xi \in [0; 4]$: $ \tilde{p}'''(\xi) = 0 $ and $ |p'''(\xi)| \geq 1 - p(3) $.
    
\end{proof}

\begin{proof}[Lemma \ref{sum_bound}]
    The Maclaurin series for the natural logarithm converges for $ -1 \leq x < 1 $. Substituting $ x = -\dfrac{1}{2} $, we can calculate and bound our sum:
    \[ \ln(1 + x) = \sum_{k = 1}^\infty \frac{(-1)^{k + 1}}{k}x^k \]
    \[ \ln{\frac{1}{2}} = \sum_{k = 1}^\infty \frac{(-1)^{k + 1}}{k}\left(-\frac{1}{2}\right)^k = -\sum_{k = 1}^\infty \frac{1}{k2^k} \quad \Ra \quad \sum_{k = 1}^\infty \frac{1}{k2^k} = \ln{2} \]
    \[ \sum_{k = 4}^{\infty} \frac{1}{k2^{k - 2}} = 4\sum_{k = 4}^{\infty} \frac{1}{k2^k} = 4\left( \sum_{k = 1}^{\infty} \frac{1}{k2^k} - \frac{2}{3} \right) = 4\left( \ln{2} - \frac{2}{3} \right) < \frac{1}{8} \]
\end{proof}

\begin{proof}[Theorem \ref{example}, property 3]
    As $ T_k(\cos{x}) = \cos{kx} $, we can see that all the roots of the Chebyshev polynomial lie on $ [-1; 1] $. By the Rolle's theorem, all the roots of any derivative of the Chebyshev polynomial also lie on $ [-1; 1] $. From \cite[lemma 5.17]{markov} we get that $ T_k'''(1) = \dfrac{k^2(k^2 - 1)(k^2 - 2)}{15} > 0 $, so $ T_k'''(x) > 0 $ for $ x \geq 1 $ and $ T_k''(\theta) \geq T_k''(1) $ for $ \theta \geq 0 $.
\end{proof}

\end{document}